\documentclass{article}

\usepackage{microtype}
\usepackage{graphicx}
\usepackage{subfigure}
\usepackage{booktabs} 

\usepackage{hyperref}

\usepackage[accepted]{icml2024}

\usepackage{amsmath}
\usepackage{amssymb}
\usepackage{mathtools}
\usepackage{amsthm}

\usepackage[capitalize,noabbrev]{cleveref}

\theoremstyle{plain}
\newtheorem{theorem}{Theorem}[section]
\newtheorem{proposition}[theorem]{Proposition}
\newtheorem{lemma}[theorem]{Lemma}

\theoremstyle{definition}

\theoremstyle{remark}

\usepackage[textsize=tiny]{todonotes}

\usepackage{amsfonts}       
\usepackage{nicefrac}       
\usepackage{microtype}      
\usepackage{xcolor}         
\usepackage{amsmath,amssymb}
\usepackage{mathrsfs}

\usepackage{multirow}
\usepackage{array}
\usepackage{adjustbox}
\usepackage{pgf}
\usepackage{paralist}

\newcommand{\E}{\mathbb{E}}

\usepackage{soul}
\usepackage{xcolor} 

\definecolor{highlight}{RGB}{255,255,0} 
\sethlcolor{highlight} 

\begin{document}

\twocolumn[
\icmltitle{
    A Multivariate Unimodality Test Harnessing the Dip Statistic of Mahalanobis Distances Over Random Projections
}



\icmlsetsymbol{equal}{*}

\begin{icmlauthorlist}
\icmlauthor{Prodromos Kolyvakis}{ovr}
\icmlauthor{Aristidis Likas}{ioa}
\end{icmlauthorlist}

\icmlaffiliation{ovr}{ORamaVR SA, Geneva, Switzerland}
\icmlaffiliation{ioa}{Department of Computer Science and Engineering, University of Ioannina, 45110 Ioannina, Greece}

\icmlcorrespondingauthor{Prodromos Kolyvakis}{prodromos.kolyvakis@oramavr.com}

\icmlkeywords{multivariate unimodality test, alpha-unimodality, dip test, random projections, clustering, cluster number estimation}

\vskip 0.3in
]



\printAffiliationsAndNotice{}  

\begin{abstract}
Unimodality, pivotal in statistical analysis, offers insights into dataset structures and drives sophisticated analytical procedures.
While unimodality's confirmation is straightforward for one-dimensional data using methods like Silverman's approach and Hartigans' dip statistic, its generalization to higher dimensions remains challenging.
By extrapolating one-dimensional unimodality principles to multi-dimensional spaces through linear random projections and leveraging point-to-point distancing, our method, rooted in $\alpha$-unimodality assumptions, presents a novel multivariate unimodality test named \textit{mud-pod}.
Both theoretical and empirical studies confirm the efficacy of our method in unimodality assessment of multidimensional datasets as well as in estimating the number of clusters.    
\end{abstract}

\section{Introduction}
\label{intro}

Unimodality, a fundamental concept in statistical analysis, serves as a critical lens through which one can decipher the inherent structure and patterns within datasets.
Understanding unimodality is paramount for multiple reasons.
Firstly, it provides a rudimentary insight into the nature of the data, highlighting whether the data points converge towards a common central tendency or deviate significantly.
Secondly, unimodality serves as a precursor to more complex analytical procedures, such as clustering algorithms, determining their necessity, and potentially influencing their outcomes \cite{dip-means,daskalakis_2013,daskalakis_2014}.
In essence, the importance of unimodality transcends mere statistical significance, extending its value to practical, real-world applications.

In one-dimensional data, unimodality can be fundamentally understood as the task of discerning whether a given distribution exhibits a single prominent peak or mode.
A notable advantage is that one-dimensional unimodality can be confirmed using robust statistical hypothesis tests,  particularly for one-dimensional data.
Methods such as Silverman's approach, exploiting fixed-width kernel density estimates \cite{silverman}, the widely recognized Hartigans' dip statistic \cite{hartigans} and more recently the UU-test \cite{chasani2022} are prime examples.

Nevertheless, when transitioning to higher dimensions, the process of defining unimodality becomes less straightforward, even when considering only symmetric distributions.
The intricacies of multi-dimensional spaces impose challenges that are not present in one-dimensional settings, leading to diverse interpretations and approaches to gauge unimodality.
Even worse, these intricacies make the generalization of unimodality tests notably challenging.
Numerous efforts have been made to capture the geometric essence of unimodality in \( \mathbb{R}^d \) ($d>1$) and translate it into an analytical framework \cite{Dai_1989}.
In a seminal work by \citet{alphaunim}, a definition of generalized unimodality characterized by a positive parameter \(\alpha\) was proposed (called $\alpha$-unimodality), which is pertinent to distributions across \(\mathbb{R}^d\) and offers a broader perspective that encompasses many aspects of 1-dimensional unimodality (\citealp[Chapter 3.2]{dharmadhikari1988unimodality}). 
Despite the various definitions, however, few methods are available for assessing unimodality in multidimensional data vectors.

In parallel, another line of research has delved into the use of random projections as a strategy for capturing the essence of multi-dimensional distributions.
Random projections, known for its efficacy in dimensionality reduction, have shown significant potential for learning mixtures of Gaussians \cite{Dasgupta1999LearningMO}.
Additionally, the Diaconis-Freedman effect elucidates the behavior of random projections of probability distributions in the high-dimensional space \cite{diaconis}.
Specifically, for a given probability distribution \( P \) in a \( d \)-dimensional space, when we consider a dimension \( q \) much smaller than \( d \), the majority of the \( q \)-dimensional projections of \( P \) resemble scaled mixtures of spherically symmetric Gaussian distributions \cite{dumbgen2013low}.
Consequently, random projections appear to be a potent tool for the analysis of unimodality as they facilitate the transformation of the problem into a seemingly simpler space.
However, not every random projection is conducive to our analysis; many projections can obfuscate distinct modes by distorting distances.
Consequently, we limit our approach to a family of random projections that, with a certain probability, maintain pairwise distances.

In this work, we tackle the complexity of extending unimodality testing to higher dimensions.
Echoing the principles of $\alpha$-unimodality, we introduce a novel algorithm for efficient multivariate unimodality testing.
Our approach bridges the gap between the simplicity of one-dimensional unimodality confirmation and the intricacies of its higher-dimensional counterpart.
Central to our investigation are the $\alpha$-unimodality preserving properties of point-to-point distancing and linear random projections.
We demonstrate that linear random projections preserve the $\alpha$-unimodality property in Mahalanobis distances from a reference point.
Leveraging these one-dimensional Mahalanobis distances, we apply the dip test for unimodality detection.
Employing various random projections, akin to Monte Carlo simulations, we assess the $\alpha$-unimodality of the original data distribution.
In summary, the contribution of our work is two-fold.
Firstly, to the best of our knowledge, we propose the first mathematically founded multivariate unimodality test, dubbed \textit{mud-pod}\footnote{Multivariate Unimodality Dip based on random Projections, Observers \& Distances.}.
Secondly, we present the \textit{mp-means} incremental clustering method which is a wrapper around k-means exploiting mudpod for unimodality assessment.
Both theoretical findings and empirical validations underpin our methods, showcasing efficacy in unimodality assessments and clustering scenarios.

\section{Related Work}
\label{related-work}

A multivariate unimodality test that aligns closely with our work is the \textit{dip-dist}, introduced by \citet{dip-means}.
This criterion aims to ascertain the modality (unimodal vs. multimodal) of a dataset by applying the unidimensional dip test on each row of the pairwise distance matrix of the dataset.
The rationale behind this approach is that by selecting an arbitrary data point and calculating its distances to all other points, we obtain a snapshot of the underlying cluster morphology.
In presence of a single cluster, the distribution of distances is expected to be unimodal.
To further the applicability of this criterion, it has been integrated into a clustering method named \textit{dip-means}.
This incremental algorithm employs cluster splitting based on the dip-dist to determine if a cluster should be divided.
Consequently, dip-means can automatically estimate the cluster count.
However, the dip-dist criterion is not without its limitations.
A glaring drawback is its reliance on pairwise distances and its operation within the original data space, which can present challenges in certain scenarios.
Moreover, the dip-dist method operates in an ad-hoc manner, lacking a rigorous mathematical foundation.
In our work, we address these limitations by introducing random projections to assess unimodality on randomly projected distances.
This approach permits Monte Carlo hypothesis testing by enabling sampling, previously not directly feasible in the original space.
Additionally, we leverage the Mahalanobis distance and we empirically demonstrate its added benefits.
Last but not least, we provide a mathematical formulation and prove the consistency of our test, thereby establishing the missing foundation for the dip-dist method.

The folding test, introduced in \cite{folding-test}, offers a versatile evaluation method suitable for both univariate and multivariate scenarios.
This test revolves around the concept of \textit{folding}, involving three key steps:
\begin{inparaenum}[(a)]
    \item folding the distribution with respect to a designated pivot \(s\),
    \item calculating the variance of the folded distribution, and lastly,
    \item comparing this folded variance to the original variance.
\end{inparaenum}
The central idea behind the folding test is that when applied to multimodal distributions, the folded distribution typically exhibits a significantly reduced variance compared to its unimodal counterparts.
A limitation of the folding test is its reliance on the empirical assumption that folding a multimodal distribution leads to a reduction in variance.
Consequently, the concept of unimodality is not explicitly integrated into the folding test computation. It is important to note that there are cases where this assumption does not hold true, resulting in incorrect outcomes for the folding test \cite{chasani2022}.
Another research direction focuses on examining particular families of unimodal distributions.
In the work of \citet{log-concavity-test}, a scalable test for log-concavity is elucidated building on maximum likelihood estimation (MLE), validated in finite samples across any dimension.
A noteworthy empirical observation from their research is the pronounced efficacy achieved by adopting random projections.
However, it is crucial to acknowledge that while log-concave functions capture a substantial subset of unimodal functions, they fall short of encompassing the entirety of the concept and may struggle to extend their applicability to more diverse and real-world scenarios.

\section{Methods}
\label{method}

\subsection{Preliminaries}
\label{preliminaries}
In the following, we present key notation and foundational background applicable throughout this paper.
We use capital letters to denote random variables or matrices and boldface type to represent vectors.

\subsubsection{\(\alpha\)-Unimodal Distrubutions}

Without loss of generality, we assume that the mode of the unimodal distribution is at \(\mathbf{0}\).
A random d-vector \(\mathbf{X} \in \mathbb{R}^d\) is said to have an \(\alpha\)-\textit{unimodal distrubution} about \(\mathbf{0}\) if, for every bounded, nonnegative, Borel measurable function \(g\) on \(\mathbb{R}^d\) the quantity \(t^\alpha\E[g(t\mathbf{X})]\) is nondecreasing in \(t \in (0, \infty)\).
In what follows, we use the notation \( \mathbf{X} \sim \mathcal{P_{\alpha}} \) to represent a d-vector following an \( \alpha \)-unimodal distribution.
It follows from the defintion that if \(\mathbf{X} \sim \mathcal{P_{\alpha}}\) and \(\alpha < \beta\), then \(\mathbf{X} \sim \mathcal{P_{\beta}}\).
An important equivalent characterization for the set of \(\alpha\)-unimodal distrubutions is the Decomposition theorem:
\( \mathbf{X} \sim \mathcal{P_{\alpha}} \) iff \(\mathbf{X}\) is distributed as \(U^{\frac{1}{\alpha}}\mathbf{Z}\) where \(U\) is uniform on \((0,1)\) and \(\mathbf{Z}\) is independent of \(U\) \cite{alphaunim}.

This theorem closely mirrors the intuition of one-dimensional unimodality, since \citet{khintchine1938unimodal} demonstrated that a real random variable \( X \) has a unimodal distribution iff \( X \sim U Z \), where \( U \) is uniform on \([0,1]\) and \( U \) and \( Z \) are independent.
It follows that a scalar \( X \sim \mathcal{P_{\alpha}} \) iff \( X^{\alpha} \) is unimodal as per the standard definition in \(\mathbb{R}\).
Next, we present and prove three pivotal properties of \(\alpha\)-unimodality, i.e., the translation, norm, and projection properties.

\begin{lemma}[Translation Property]
\label{trans_prop}
Let \( \mathbf{X} \sim \mathcal{P_{\alpha}} \) and \(\mathbf{c} \in \mathbb{R}^d\), then \( \mathbf{X} + \mathbf{c} \sim \mathcal{P_{\alpha}} \) .
\end{lemma}
\begin{proof} Let \(t^\alpha\E[g(t\mathbf{X} +t\mathbf{c})] = t^\alpha\E[h(t\mathbf{X})]\), where \(h(\mathbf{x}) = g(\mathbf{x}+\mathbf{c})\).
Note that \(h\) is bounded, nonnegative, Borel measurable and that the first expression is nondecreasing iff the last expression is nondecreasing in \(t\).
\end{proof}

\begin{lemma}[Norm Property]
\label{norm_prop}
If \( \mathbf{X} \sim \mathcal{P_{\alpha}} \), then \( ||\mathbf{X}|| \sim \mathcal{P_{\alpha}} \).
\end{lemma}
\begin{proof}
\(
||\mathbf{X}|| = \sqrt{(U^{\frac{1}{\alpha}}\mathbf{Z})^\top (U^{\frac{1}{\alpha}}\mathbf{Z})} = U^{\frac{1}{\alpha}} ||\mathbf{Z}||
\)
\end{proof}
\begin{lemma}[Projection Property]
\label{proj_prop}
Let \( \mathbf{X} \sim \mathcal{P_{\alpha}} \) and let \(A\) be a real matrix from \(\mathbb{R}^d\) to \(\mathbb{R}^q\), then \( A\mathbf{X} \sim \mathcal{P_{\alpha}} \).
\end{lemma}
\begin{proof}
A direct use of the Decomposition Theorem.
\end{proof}

Leveraging the foundational properties described above, we now present a salient result.
This result not only fortifies our understanding of \(\alpha\)-unimodal distributions but will also play an instrumental role in the rest of the paper.

\begin{lemma}[Mahalanobis]
\label{mahalanobis}
Let \( \mathbf{X} \sim \mathcal{P_{\alpha}} \) with a well-defined covariance matrix \(\Sigma\) and \(\mathbf{o} \in \mathbb{R}^d\), then the distribution of the Mahalanobis distances with respect to \(\mathbf{o}\), given by
\(
\sqrt{(\mathbf{X}-\mathbf{o})^\top{\Sigma}^{-1}(\mathbf{X}-\mathbf{o})}
\)
is \(\alpha\)-unimodal.
\end{lemma}
\begin{proof}
Given a positive semidefinite covariance matrix \(\Sigma\), and utilizing the matrix square root decomposition, the Mahalanobis distance can be expressed as:
\begin{align*}
\sqrt{(\mathbf{x}-\mathbf{o})^\top\Sigma^{-1}(\mathbf{x}-\mathbf{o})} &= ||{\Sigma}^{-\frac{1}{2}}(\mathbf{x}-\mathbf{o})||
\end{align*}
Proof stems from translation and projection lemmas.
\end{proof}

Mahalanobis distance possesses distinct properties: it is unitless, scale-invariant, and considers the covariance structure across all dimensions.
Traditionally, it has been employed in multivariate hypothesis testing.
Notably, the Hotelling's $T^2$ statistic \cite{t-squared}, which generalizes the Student's t-statistic, exemplifies its usage.
The Mahalanobis distance is pivotal to our multivariate unimodality test.

\subsubsection{Dip Test}

The dip test serves as a tool for discerning multimodality within a unidimensional dataset.
It gauges this by examining the maximum deviation, i.e., the Kolmogorov-Smirnov statistic, between the empirical cumulative distribution function (e.c.d.f.), $F(t)$, and the nearest unimodal c.d.f., $G(t)$.
The dip statistic for a distribution function \( F \) is defined as:
\( \text{dip}(F) = \min_{G \in U} \rho(F, G) \),
where \( \rho(F, G) = \max_t |F(t) - G(t)| \) and \( U \) represents the set of all possible unimodal distributions.
The dip test's significance is highlighted by its ability to unveil the least among the most substantial deviations between the empirical cumulative distribution function \( F \) of the univariate dataset and the c.d.f.s of the class of unimodal distributions.
A salient attribute of the dip statistic is its convergence as the sample size burgeons, such that \( \lim_{{n \to \infty}} \text{dip}(F_n) = \text{dip}(F) \) \cite{hartigans}.
Moreover, the class of uniform distributions \( U \) is acclaimed to be the most fitting for the null hypothesis, owing to its stochastically larger dip values compared to other unimodal distributions.
To calculate the dip value, the e.c.d.f. of the data is considered, and the unimodal piecewise linear function with the smallest maximum distance to the e.c.d.f. is determined.
The p-value for unimodality, derived via bootstrap samples, functions as a determinant for the dataset's modality.
A dataset with a p-value greater than \( a \) indicates unimodality; otherwise, multimodality is suggested.

\subsubsection{Random Projections}

The Diaconis-Freedman effect can be a valuable tool for unimodality analysis, simplifying the problem by likely transforming it into a Gaussian mixture model.
When considering a probability distribution \( P \) in a \( d \)-dimensional space, most \( q \)-dimensional projections of \( P \) with \( q \ll d \) resemble scale mixtures of spherically symmetric Gaussian distributions.
Additionally, linear random projections can preserve distances when projecting high-dimensional points into lower-dimensional spaces.
This phenomenon is encapsulated in the celebrated \citeauthor{Johnson1984ExtensionsOL} lemma presented below \cite{Johnson1984ExtensionsOL,jl_granda}:

\textbf{JL Lemma}:
Let \(S := \{x_i\}_{i=1}^{k}\) be a subset of \(\mathbb{R}^d\) and \(\epsilon > 0\). Then, let \(\Pi \in \mathbb{R}^{d \times q}\), where \(q \geq \mathit{8 \log(k)/\epsilon^2}\), be a random matrix with i.i.d. entries \(\Pi_{ij} \sim \mathcal{N}(0, \mathit{1/d)}\). With probability at least \(\mathit{1/k}\), for any \(x_i, x_j \in S\), we have:
\[
(1 - \epsilon) \| x_i - x_j \|^2 \leq \| \Pi x_i - \Pi x_j \|^2 \leq (1 + \epsilon) \| x_i - x_j \|^2.
\]
We denote by \(\mathcal{R}_{\mathbf{\Pi}}\) the set of matrices fulfilling the distance preservation criteria specified in the JL Lemma.
According to the JL Lemma, if \(\Pi\) is sampled with i.i.d. \(\mathcal{N}(0, \frac{1}{d})\) entries, then \(P(\Pi \in \mathcal{R}_{\Pi}) \geq \frac{1}{k}\).
By employing a square root decomposition, we can demonstrate the applicability of the JL lemma to the Mahalanobis distance \cite{rand_proj_mahal}, while
mitigating the singularity issue inherent in inverting the covariance matrix in high dimensions \cite{wainwright_rpt,raptt}.

\subsection{Connecting the Dots}

Given a dataset of multidimensional data vectors, assessing its unimodality becomes intricate.
Random projections offer a solution by maintaining key pairwise distances and performing unimodality assessment in a more Gaussian-like space.
By picking an arbitrary \textit{observer} data point and deriving its distances to all other points, we garner a snapshot of the underlying cluster morphology.
In presence of a single cluster, the distribution of distances is proven to be unimodal.
Notably, the narrative this observer presents is contingent upon its location.
We integrate this idea of random projection and the observer’s perspective into what we term a \textit{view}.
Our proposed algorithm focuses on analyzing these views, pinpointing those views that contradict unimodal narratives, and thus highlight multifaceted cluster formations.
In the rest of the section, we will rigorously define the aforementioned concept.

Given a set $\mathcal{S}_{\mathbf{X}}$ of points from $\mathbf{X} \sim \mathcal{P_{\alpha}}$, let $\mathbf{o} \in \mathcal{S}_{\mathbf{X}}$ be a random point, dubbed observer.
We define the set of Mahalanobis distances with regard to this observer as follows:
\[
\mathcal{D}_{\mathbf{S}}^{o}  = \left \{ ||{\mathbf{\Sigma}}^{-\frac{1}{2}}(\mathbf{x}-\mathbf{o})|| \text{ }| \text{ } \mathbf{x} \in \mathcal{S}_{\mathbf{X}}\setminus\{\mathbf{o}\} \right \}.
\]
Let $\Pi \in \mathcal{R}_{\mathbf{\Pi}}$, we define 
$\Pi \circ \mathcal{D}_{\mathbf{S}}^{o}$ to be the set of the Mahalanobis distances of the randomly projected points with respect to an observer $\mathbf{o}$. Specifically, we have:
\[
\Pi \circ \mathcal{D}_{\mathbf{S}}^{o}  = \left \{ ||{\mathbf{\Sigma}}^{-\frac{1}{2}}_{\Pi}\Pi(\mathbf{x}-\mathbf{o})|| \text{ }| \text{ } \mathbf{x} \in \mathcal{S}_{\mathbf{X}}\setminus\{\mathbf{o}\} \wedge \Pi \in \mathcal{R}_{\mathbf{\Pi}} \right \},
\]
where \(\mathbf{\Sigma_{\Pi}} = \Pi\mathbf{\Sigma}\Pi^T\).
It is important to note that since \( \mathbf{X} \sim \mathcal{P_{\alpha}} \), the elements of both \(\mathcal{D}_{\mathbf{S}}^{o}\) and \( \Pi \circ \mathcal{D}_{\mathbf{S}}^{o} \) exhibit \(\alpha\)-unimodal distributions, as established earlier. This yields the subsequent observation.

\begin{proposition}[Randomisation Hypothesis]
Given \( \mathbf{X} \sim \mathcal{P}_{\alpha} \), the distributions of elements within \(\mathcal{D}_{\mathbf{S}}^{o}\) and \( \Pi \circ \mathcal{D}_{\mathbf{S}}^{o} \) retain \(\alpha\)-unimodality under any transformation \( \Pi \in \mathcal{R}_{\mathbf{\Pi}}\).
\end{proposition}

The Randomisation Hypothesis (RH) is central to our analysis. RH facilitates the execution of a series of one-dimensional unimodality tests, subsequently allowing the evaluation of the \(\alpha\)-unimodality of the distribution that produces our data \(X\).
Under the RH, every randomly projected distances should exhibit unimodality.
Any deviation from this expected behavior can signal a departure from unimodality in the original data distribution.
Random projections confer several distinct merits.
Firstly, they preserve the pairwise distances, as endorsed by the JL lemma.
Secondly, given our observations, random projections serve as an invaluable tool for unimodality investigation.
They transmute the challenge into a space resembling a mixture of Gaussians.
Furthermore, they ameliorate the singularity problem associated with the inversion of the covariance matrix used by the Mahalanobis distance.
Lastly, they pave the way for harnessing Monte Carlo simulation for hypothesis testing \cite{lehmann2005testing}, i.e., they establish the bedrock for sampling from a distribution, specifically \( \mathcal{R}_{\mathbf{\Pi}} \), that is ostensibly simpler than the original data distribution of \(\mathbf{X}\).

\subsection{Multivariate Unimodality Testing}

Building on the aforementioned foundation, we now introduce our multivariate \(\alpha\)-unimodality test called \emph{mud-pod}.
For a given \(\alpha\) and a set of points \(\mathcal{S}_{\mathbf{X}}\) from \(\mathbf{X} \sim \mathcal{P_{\alpha}}\), we define our hypotheses:
\begin{align*}
H_0 &: \mathbf{X} \sim \mathcal{P_{\alpha}} \quad \text{vs.} \quad H_1 : \mathbf{X} \not\sim \mathcal{P_{\alpha}}.
\end{align*}
We define the pairing of a random projection \(\Pi\) with an observer \(\mathbf{o}\) as a random view.
We assume independence between the random projection $\Pi$ and the observer $\mathbf{o}$.
Given a set of $N$ random vectors $\mathcal{S}$ and a random view, we can obtain the corresponding set of Mahalanobis distances \(\Pi \circ \mathcal{D}_{\mathbf{S}}^{o}\).
Under the null hypothesis, recall that \(\Pi \circ \mathcal{D}_{\mathbf{S}}^{o} = \{ d_i \}_{i=1}^{N-1}\) is \(\alpha\)-unimodal, and the set \( \{ d_i^{\alpha} \}_{i=1}^{N-1} \) is unimodal, allowing the employment of the dip test.
Let \(T(\Pi \circ \mathcal{D}_{\mathbf{S}}^{o})\) denote the dip test p-value.
If \(a \in [0,1]\) is the significance level, the null hypothesis is rejected iff \(T(\Pi \circ \mathcal{D}_{\mathbf{S}}^{o}) \leq a\).

Utilizing the idea of random views, which, as previously discussed, preserve \(\alpha\)-unimodality, we can employ them as a foundation for Monte Carlo simulations.
The rationale is that as more views reject the null hypothesis, our confidence about data multimodality increases.
Let \( \{ \Pi_i \}_{i=1}^{M} \subset \mathcal{R}_{\mathbf{\Pi}}\) be a set of $M$ random projections.
Leveraging Monte Carlo hypothesis testing theory (\citealp[Chapter 11.2.2]{lehmann2005testing}) and building on the fact that the dip statistic has a well defined c.d.f. \cite{hartigans}, we explore the conditional c.d.f. of the dip test: \(J_{N}(t) = P\left (T(\Pi \circ \mathcal{D}_{\mathbf{S}}^{o}) \leq t \mid \mathcal{S}_{\mathbf{X}})\right)\).
The aforementioned probability is measured over the joint distribution of random projections \(\mathcal{R}_{\mathbf{\Pi}}\) and the distribution of observers \(\mathcal{O}\).
Let $I\{.\}$ denote the indicator function, we define $\hat{J}_{n, M}(t)$ as the approximation of the true c.d.f $J_{N}(t)$ computed on the series of the random views:
\[
    \hat{J}_{N, M}(t) = M^{-1} \sum_{i=1}^{M} I \left \{T(\Pi \circ \mathcal{D}_{\mathbf{S}}^{o}) \leq t \right \}
\]
By a direct application of the Glivenko-Cantelli theorem, we have that \(\hat{J}_{N, M}(t)\) converges w.p. 1 to \(J_{N}(t)\) \cite{sharipov2011}.
Interestingly, the Dvoretsky, Kiefer, Wolfowitz inequality provides bounds on the closeness between \( \hat{J}_{N, M}(t) \) and \( J_{N}(t) \) for a given \( M \).
Specifically, we have:
\[
P \left ( \sup_{t\in \mathbb{R}} | \hat{J}_{N, M}(t) - J_{N}(t) | > \tau \right ) \leq C e^{-2M\tau^2}
\]
Hitherto, we have not delineated the methodology for selecting observers from the set of points obtained from a random projection. Several sampling strategies can exist, with the most intuitive being uniform random sampling.
However, empirical results suggest that uniformly selecting observers based on a specific percentile of the distance distribution from the samples' mean yields superior performance.
The underlying rationale is that points situated farther away from the means possess a better capability to discern the topographical elevations formed by distinct clusters \cite{dip-means}.
It is important to note that despite the dependency introduced between the observer $\mathbf{o}$ and the random projection $\Pi$ by the \textit{percentile strategy}, our analysis remains valid thanks to the extension of the Glivenko-Cantelli theorem to strictly stationary sequences \cite{sharipov2011}.
Ultimately, the projection dimension is the minimal integer that satisfies the JL lemma for a specified \(\epsilon\).
Algorithm~\ref{alg:mud-pod} details the complete mud-pod test.
\begin{algorithm}
\caption{mud-pod ($\alpha$, $X$, $a$, $M$, $p$, $\epsilon$)}
\textbf{Input:} $\alpha$ (the positive unimodality index), $X$ (a set of real vectors), $a$ (a significance threshold), $M$ (number of simulations), $p$ (p-th percentile), $\epsilon$ (distance distortion)

\textbf{Output:} p-value of the test

\begin{algorithmic}[1]
\FOR{$i = 1$ to $M$}
    \STATE Project the points via a \( \left\lceil \frac{8 \log(|X|)}{\epsilon^2} \right\rceil \) random projection, resulting in $X\Pi_i$.
    \STATE Select an observer $\textbf{o}$ from the p-th percentile of the projected Mahalanobis distances from the mean.
    \STATE Compute the set of distances from $o$, i.e., $\Pi \circ \mathcal{D}_{\mathbf{S}}^{o}$.
    \STATE Conduct a dip test on exponentiated $\Pi \circ \mathcal{D}_{\mathbf{S}}^{o}$ distances.
\ENDFOR
\STATE \textbf{return}  \( p_{\text{value}} = \frac{1}{M} \sum_{i=1}^{M} I\{T(\Pi_i \circ \mathcal{D}_{\mathbf{S}}^{o}) \leq a\} \)
\end{algorithmic}
\label{alg:mud-pod}
\end{algorithm}

\section{Experiments}
\label{exps}

In this section, we present a comprehensive suite of experiments conducted for both multivariate unimodality testing and estimating cluster counts in clustering tasks.
Our decision to assess our algorithm for cluster estimation is driven by its complexity and wide practical relevance \cite{stop_using_elbow}.
A pertinent query pertains to which \(\alpha\)-unimodality family we aim to detect.
Despite its rigor, we opted to assess 1-unimodality regardless of the underlying data space.
Empirical results indicated its efficacy even on challenging real-world datasets.
Our algorithm is characterized by three parameters: $M, \epsilon, p$.
Following an initial exploration, we identified parameter values that consistently produced favorable outcomes.
Specifically, we set \( M = 100 \), \( \epsilon = 0.99 \), \( p = 0.99 \), and chose a significance level \( a = 0.01 \).

\subsection{Unimodality Experiments}

\renewcommand{\arraystretch}{1.2}
\begin{table*}[htb]
    \centering
    \caption{Performance comparison of dip-dist (DD), mudpod (MP), and folding (F) tests in determining unimodality or multimodality. The table displays the \emph{percentage of multimodality} cases identified over 10 runs. 1000 points randomly sampled from synthetic sets and MNIST training set per experiment. For space constraints, Single MNIST experiment results are compressed, with digits divided by semi-colons summarizing the outcomes across all tests. \( G_{n} \) denotes a \(n\)D Gaussian distribution. \( C(r) \) symbolizes the 2D equation of a circle with radius \( r \). Utilizing parametric equations \( U(\theta) = \left( \cos(\theta), \sin(\theta) \right) \) and \( L(\theta) = \left( 1 - \cos(\theta), 1 - \sin(\theta) - 0.5 \right) \) with \( \theta \in [0, \pi] \).
    \vspace{0.1in}
}
    \begin{tabular}{c|c|c|c|c}
    \hline
    Experiment & Distribution Details & DT & \textbf{MP} & F \\
    \hline
     Single 2D Gaussian & \(G\left([0, 0], \mathrm{I}\right)\) & 0 & 0 & 0 \\
     Single 3D Gaussian & \(G\left([0, 0, 0], \mathrm{I}\right)\) & 0 & 0 & 0 \\
      Two 2D Circles & \( \frac{1}{2} \left( C(0.5) + \mathcal{N}(0, 0.05^2I) \right) + \frac{1}{2} \left( C(1) + \mathcal{N}(0, 0.05^2I) \right) \) & 100 & 100 & 100 \\
     Two 2D Moons & \( \frac{1}{2} \left( U(\theta) + \mathcal{N}(0, 0.05^2I) \right) + \frac{1}{2} \left( L(\theta) + \mathcal{N}(0, 0.05^2I) \right)
 \) & 100 & 100 & 0 \\
     Two 2D Gaussians & \(0.5 \cdot G_1\left([1, 4], \mathrm{I}\right) + 0.5 \cdot G_2\left([2, 1], \mathrm{I}\right)
\) & 100 & 100 & 0 \\
     Three 2D Gaussians & \(\frac{1}{3} \cdot G_1\left([t, t],  \mathrm{I}\right) + \frac{1}{3} \cdot G_2\left([0, 0],  \mathrm{I}\right) + \frac{1}{3} \cdot G_3\left(-[t, t],  \mathrm{I}\right) | \text{ } t=2.5
\) & 50 & 100 & 0 \\
     Two 3D Gaussians & \(0.5 \cdot G_1\left([1, 4, 2], \mathrm{I}\right) + 0.5 \cdot G_2\left([1, -2, 3], \mathrm{I}\right)
\) & 100 & 100 & 0 \\
     Three 3D Gaussians & \( \frac{1}{3} \cdot G_1\left([t, t, t], \mathrm{I}\right) + \frac{1}{3} \cdot G_2\left([0, 0, 0], \mathrm{I}\right) + \frac{1}{3} \cdot G_3\left(-[t, t, t],\mathrm{I}\right) | \text{ } t=2.9\) & 10 & 100 & 0 \\
    \hline
     Single Digit MNIST & 0; 2; 3; 4; 7; 8 & 0 & 0 &  100 \\
     Single Digit MNIST & 1 & 100 & 100 &  100 \\
     Single Digit MNIST & 5 & 0 & 10 &  100 \\
     Single Digit MNIST & 6 & 10 & 10 &  100 \\
     Single Digit MNIST & 9 & 10 & 20 &  100 \\
     Even Digits MMNIST & \(\{0, 2, 4, 6, 8\}\) & 10 & 90 & 100 \\
     Odd Digits MMNIST & \(\{1, 3, 5, 7, 9\}\) & 80 & 100 & 100 \\
     All Digits MMNIST & \(\{0, 1, 2, 3, 4, 5, 6, 7, 8, 9\}\) & 0 & 100 & 100 \\
    \hline
    \end{tabular}
    \label{tab:unim_exps}
\end{table*}
\renewcommand{\arraystretch}{1}
Table~\ref{tab:unim_exps} presents an intricate assessment of the capability of three distinct tests — dip-dist (DD), mudpod (MP), and folding (F) — in discerning unimodal and multimodal datasets.
DD and F tests are non-parametric and we also set $a=0.01$.
It is important to note that DD is a special case of mud-pod, omitting \( \alpha \) exponent, operating in the original space using the Euclidean distance and considering all data points as observers.
The evaluation was carried out over ten distinct runs for each test on a combination of both synthetic and real-world data drawn from the MNIST dataset \cite{lecun1998mnist}.
Starting with synthetic unimodal datasets, namely the single 2D and 3D Gaussian distributions, we find a unanimous agreement across the three tests, with none indicating any instances of multimodality.
Examining the synthetic bimodal distributions, 2D Moons and Circles show clear multimodality, confirmed by DD and MP tests with 100\% detection.
DD and MP also report 100\% detection for bimodal Gaussians in 2D and 3D.
However, DD's performance declines with three closely aligned Gaussians in both 2D and 3D, unlike MP's consistent 100\% detection.
The F test only identifies multimodality in the 2D Circles dataset.

Transitioning to real-world datasets, such as MNIST, offers a more intricate perspective.
For our tests, we utilized the flattened MNIST representations without any transformations.
It is noteworthy that, while MNIST has labels, their alignment with clustering in the original space is not guaranteed.
For instance, consider the digit ``1'', representable as a single stroke or a combination of two distinct ones.
It can be observed that for digits 0, 2, 3, 4, 7, and 8, the DD and MP tests consistently reject multimodality, while the F test falsely decides multimodality in all cases.
Notably, all tests unanimously flag digit 1 as multimodal.
The digits 5, 6, and 9 reveal varied outcomes, with the F test persistently recognizing multimodality.
The conduction of three additional tests to assess multimodality of even, odd, and all MNIST subsets in multi-digit scenarios reveals a decline in DD's efficacy.
In summary, we observe a pronounced tendency of the F test to detect multimodality across various datasets, with DD and MP performance being comparable in simpler scenarios, but DD falling short in multi-digit scenarios.

\begin{table}[t]
\centering
    \caption{Impact of space, distance, and observer selection strategy on mudpod's unimodality detection performance. Mudpod's result agreement is shown for a mixture of two 2D Gaussians with confirmed ground truth unimodality. \emph{Notation:} `O' for Original, `RP' for Randomly Projected, `E' for Euclidean, `M' for Mahalanobis; `R' and `P' denote Random and Percentile, respectively.
    \vspace{0.1in}
    }
    \begin{tabular}{c|c|c|c}
    \hline
    Space & Distance     & Observer  & Agreement (\%) \\
    \hline
    O     & E    & R    & 0.80 \\
    O     & E    & P & 0.87 \\
    O     & M  & R    & 0.82 \\
    O     & M  & P & 0.87 \\
    RP    & E    & R    & 0.85 \\
    RP    & E    & P & 0.90 \\
    RP    & M  & R    & 0.92 \\
    RP    & M  & P & 0.95 \\
    \hline
    \end{tabular}
    \label{tab:ablation}
\end{table}
\textbf{Ablation Study:} Our test comprises several components, including random projections, Mahalanobis distance, and uniform sampling of distances from origin percentiles.
In Table~\ref{tab:ablation}, we present an ablation study to evaluate the significance of each component.
We have generated datasets from a mixture of two 2D Gaussians as the target distribution, for which the bimodality or unimodality can be determined analytically \cite{konstantellos}.
The reported performance is aggregated from a series of experiments conducted with four different significance levels \(0.001, 0.005, 0.01, 0.05\), across 1000 distinct data sets, with each experiment executed 10 times.
Our ablation study reveals key insights into the performance of different observer picking strategies, spaces, and distances for unimodality detection.
Primarily, the percentile strategy (P) for observer picking demonstrated superiority over the random strategy (R) across all tested combinations of space and distance.
Furthermore, the Mahalanobis distance consistently emerged as a more effective metric compared to the Euclidean distance.
This performance difference was especially evident in the Randomly Projected space, suggesting that the intrinsic characteristics of the Mahalanobis distance, e.g., accounting for data covariance, plays a pivotal role in enhancing detection reliability.
Notably, the randomly projected (RP) space exhibited a pronounced advantage over the original space in our assessments.
This superiority held true irrespective of the distance metric or observer strategy employed.
Such a trend strongly indicates that the RP space aligns more coherently with the ground truth unimodality, offering better detection capabilities compared to the original space.
In summary, our findings recommend a strategic combination of the P strategy, Mahalanobis distance, and the RP space.

\begin{table*}[t]
\centering
\caption{The table presents the number of clusters ( K ) and the associated NMI values obtained by various methods on different datasets. Values marked with † could not be computed due to memory constraints or were terminated after 8 hours. All results are represented as mean ± standard deviation over 10 executions. For the k-means algorithm, the correct number of clusters was always predefined.
\vspace{0.1in}
}
\begin{adjustbox}{center}
\begin{tabular}{l|cc|cc|cc|cc}
\toprule
\multirow{2}{*}{Method} & \multicolumn{2}{c|}{USPS} & \multicolumn{2}{c|}{MNIST} & \multicolumn{2}{c|}{F-MNIST} & \multicolumn{2}{c}{HAR}  \\
\cmidrule(lr){2-9}
                        & k & NMI & k & NMI & k & NMI & k & NMI  \\
\midrule
Ground truth & 10 & 1.0 & 10 & 1.0 & 10 & 1.0 & 5 & 1.0 \\
\midrule
k-means & - & 0.61$\pm$0.00 & - & 0.49$\pm$0.00 & - & 0.51$\pm$0.00 & - & 0.61$\pm$0.01 \\
x-means & 35$\pm$0 & 0.61$\pm$0.01 & 35$\pm$0 & 0.55$\pm$0.00 & 35$\pm$0 & 0.51$\pm$0.00 & 41$\pm$2 & 0.56$\pm$0.01 \\
g-means & 35$\pm$0 & 0.61$\pm$0.00 & 35$\pm$0 & 0.55$\pm$0.00 & 35$\pm$0 & 0.51$\pm$0.00 & 931$\pm$29 & 0.42$\pm$0.00 \\
pg-means & 2$\pm$1 & 0.14$\pm$0.07 & 2$\pm$1 & 0.18$\pm$0.09 & 4$\pm$2 & 0.31$\pm$0.11 & 2$\pm$1 & 0.14$\pm$0.05 \\
dip-means & 4$\pm$0 & 0.44$\pm$0.00 & 1$\pm$0 & 0.01$\pm$0.05 & 9$\pm$2 & 0.50$\pm$0.01 & 3$\pm$0 & 0.73$\pm$0.00 \\
hdbscan  & 13$\pm$0 & 0.38$\pm$0.00 & 36$\pm$0 & 0.33$\pm$0.00 & 3$\pm$0 & 0.05$\pm$0.00 & 3$\pm$0 & 0.52$\pm$0.00 \\
fold-means & 31$\pm$0 & 0.59$\pm$0.01 & 31$\pm$0 & 0.55$\pm$0.01 & 31$\pm$0 & 0.54$\pm$0.01 & 11$\pm$0 & 0.61$\pm$0.00 \\
\textbf{mp-means} & 8$\pm$2 & 0.62$\pm$0.03 & 9$\pm$2 & 0.55$\pm$0.04 & 9$\pm$1 & 0.54$\pm$0.01 & 3$\pm$0 & 0.73$\pm$0.00 \\
\bottomrule
\toprule
\multirow{2}{*}{Method} &  \multicolumn{2}{c|}{Optdigits} & \multicolumn{2}{c|}{Pendigits} & \multicolumn{2}{c|}{Isolet} & \multicolumn{2}{c}{TCGA} \\
\cmidrule(lr){2-9} 
                        & k & NMI & k & NMI & k & NMI & k & NMI \\
\midrule
Ground truth  & 10 & 1.0 & 10 & 1.0 & 26 & 1.0 & 5 & 1.0 \\
\midrule
k-means & - & 0.69$\pm$0.01 & - & 0.69$\pm$0.01 & - & 0.73$\pm$0.01 & - & 0.80$\pm$0.01 \\
x-means & 35$\pm$0 & 0.71$\pm$0.01 & 35$\pm$0 & 0.70$\pm$0.01 & 233$\pm$4 & 0.66$\pm$0.01 & 20$\pm$1 & 0.68$\pm$0.01 \\
g-means & 35$\pm$0 & 0.72$\pm$0.01 & 35$\pm$0 & 0.70$\pm$0.01 & 101$\pm$6 & 0.69$\pm$0.01 & 283$\pm$48 & 0.49$\pm$0.04 \\
pg-means & 1$\pm$0 & 0.02$\pm$0.07 & 3$\pm$1 & 0.34$\pm$0.18 & † & † & 1$\pm$0 & 0.02$\pm$0.04 \\
dip-means & 1$\pm$0 & 0.00$\pm$0.00 & 16$\pm$1 & 0.71$\pm$0.02 & 4$\pm$0 & 0.44$\pm$0.01 & 2$\pm$0 & 0.50$\pm$0.01 \\
hdbscan  & 21$\pm$0 & 0.71$\pm$0.00 & 38$\pm$0 & 0.72$\pm$0.00 & 4$\pm$0 & 0.04$\pm$0.00 & 7$\pm$0 & 0.75$\pm$0.00 \\
fold-means & 4$\pm$1 & 0.45$\pm$0.11 & 1$\pm$0 & 0.00$\pm$0.00 & 1$\pm$0 & 0.00$\pm$0.00 & † & † \\
\textbf{mp-means} & 8$\pm$1 & 0.67$\pm$0.06 & 14$\pm$1 & 0.70$\pm$0.01 & 20$\pm$7 & 0.63$\pm$0.14 & 6$\pm$1 & 0.95$\pm$0.03 \\
\bottomrule
\end{tabular}
\end{adjustbox}
\label{tab:cluster_exps}
\end{table*}

\begin{figure*}[h]
    \centering
        \includegraphics[width=\textwidth]{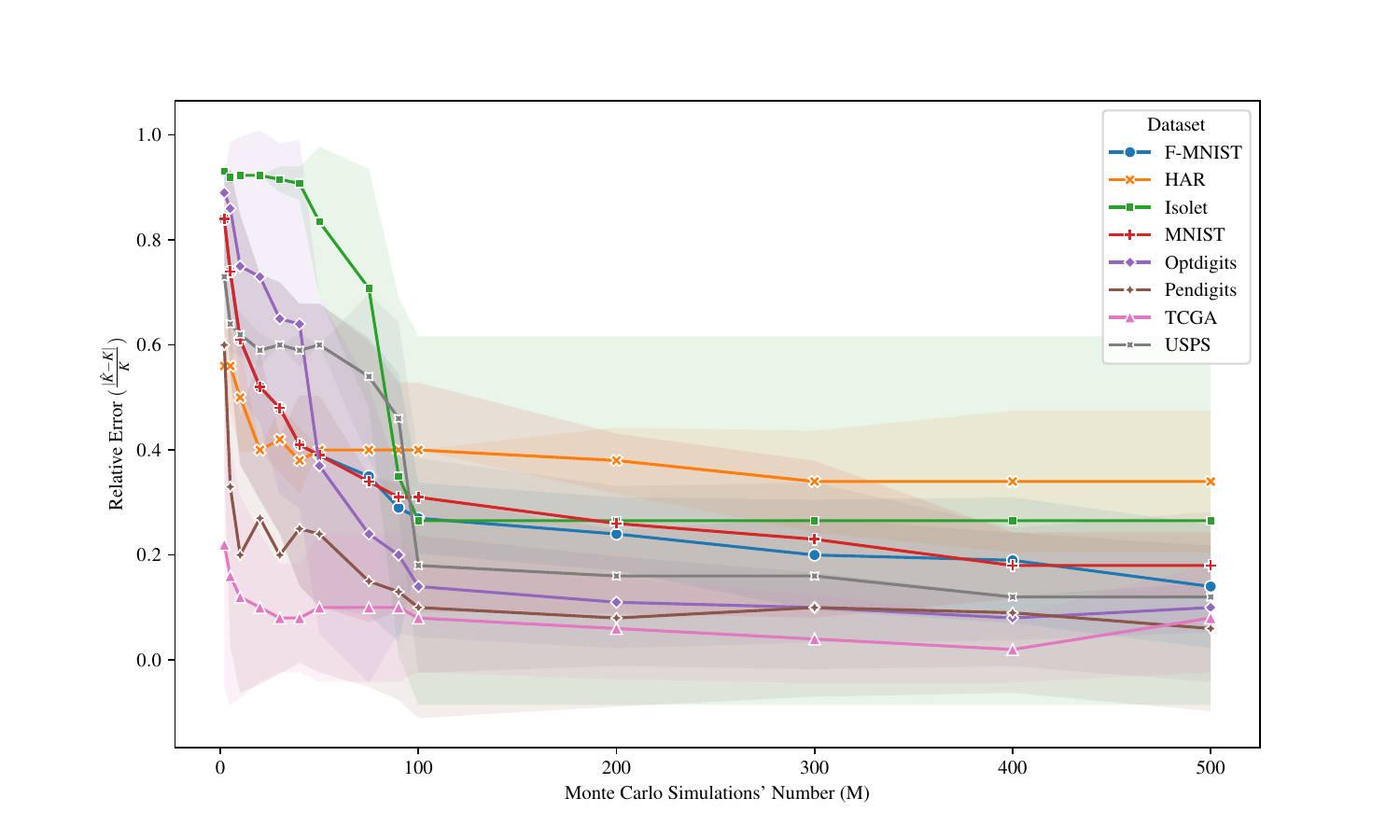}
    \caption{
        The plot shows the relative error between the estimated and actual number of clusters generated by mp-means, against increasing Monte Carlo simulations. Ten executions per experiment, variance depicted. The plot is more discernible in color.
    }
    \label{fig:relative_k_diagram}
\end{figure*}

\subsection{Clustering Experiments}
\label{res:clustering}

This section analyzes the performance of \textit{mp-means}, an approach that incorporates mud-pod into the dip-means wrapper method, replacing dip-dist.
Specifically, both mp-means and dip-means employ incremental $k$-means clustering based on testing clusters for unimodality.
Starting with one cluster (the entire dataset), they incrementally increase $k$ by splitting multimodal clusters, terminating when all clusters are deemed unimodal.
They differ in the cluster unimodality assessment.
Dip-means uses the dip-test criterion, while mp-means uses our proposed unimodality test.
Upon detecting a multimodal cluster, they split the one with the highest dip value.
It’s split into two using the 2-means algorithm or assigning clusters at mean $\pm$ standard deviation, using the cluster's mean and standard deviation.
In our work, we opt for the latter for computational efficiency.
In this way, the number of clusters is increased to $k+1$ and the $k+1$ centers are updated via $k$-means.
A similar integration of the folding test yields the \textit{fold-means} algorithm.

In Table~\ref{tab:cluster_exps}, we compare the performance of various clustering methods estimating the number of clusters across several datasets.
It is important to note that in all our experiments, we use the raw flattened data encoding and apply only a feature-wise z-transformation.
Although our method employs the Mahalanobis distance and is scale-agnostic, we observed that scaling profoundly affects other algorithms. 
Our experiments include datasets like USPS, MNIST, Fashion-MNIST (F-MNIST), Human Activity Recognition (HAR), Optdigits, Pendigits, Isolet, and TCGA.
Detailed dataset descriptions are provided in Appendix~\ref{appendix:clust_datasets}. As ground truth number of clusters, the number of classes was considered in all datasets.
We benchmark against classic algorithms, including x-means \cite{x-means}, g-means \cite{x-means}, pg-means \cite{pg-means}, dip-means \cite{dip-means}, and hdbscan \cite{hdbscan}.
A detailed overview of the baseline algorithms and their hyperparameter setups are provided in Appendix~\ref{appendix:clust_algos}.
Our evaluation metrics consist of the estimated number of clusters (k) and the Normalized Mutual Information (NMI).
NMI values lie between [0, 1], where 1 signifies a perfect match and 0 represents an arbitrary result.

This experiment set, encompassing both the complexity of many modes and inherent clustering errors, is notably more challenging than unimodality testing.
We observe that most of the classical approaches, notably x-means and g-means, tend to predict a high value for \( k \), frequently estimating around 35 clusters for datasets such as USPS, MNIST, and F-MNIST.
Similarly, x-means, consistently overestimates the \( k \) value on these datasets, accompanied by NMI values not consistently reaching the top tier.
In contrast, pg-means often significantly underestimates the number of clusters, suggesting a notable diversion from the ground truth.
This assertion is further supported by its suboptimal NMI values across multiple datasets, with values as low as 0.14 for the USPS dataset.
Hdbscan shows a varying range in \( k \) values across datasets, highlighting its adaptability.
However, this variability does not always correlate with high NMI values.
Its performance on the USPS dataset, where it estimates 13 clusters with an NMI of 0.38, is a case in point.

Turning our attention to the HAR dataset, the dip-means and mp-means methods are particularly notable.
They estimate \( k \) values closely aligned with the ground truth of 5 and simultaneously achieve the highest NMI scores, specifically 0.73.
This is indicative of their capability in deciphering the cluster structure inherent to the data. 
In the context of the Optdigits and Pendigits datasets, dip-means struggles to align its \( k \) estimations with the ground truth, predicting values of 1 and 16, respectively.
Notably, mp-means appears more consistent, with \( k \) values of 8 and 14 for Optdigits and Pendigits, respectively, and corresponding NMI scores that are commendable.
In an experiment with k-means using accurate cluster count, the NMI of mp-means closely aligns with, or even exceeds, that of k-means.
Intriguingly, the projection dimension of mp-means varies between subclusters.
The experiments also underscore its robustness across multiple projection spaces.
In conclusion, the consistent performance of mp-means underscores its utility in solving clustering problems with unknown number of clusters.

\textbf{Ablation Study on Simulations' Number:} In Figure~\ref{fig:relative_k_diagram}, the relative error between the estimated number of clusters \( K \) and the true value, as deduced by the mp-means algorithm, is depicted against the number of Monte Carlo simulations (\(M\)).
The shaded regions around each line represent the variance over 10 executions, highlighting result consistency.
As observed, for all datasets, an increase in the number of Monte Carlo simulations tends to correspond with a decline in the relative error, signifying an enhancement in the accuracy of \( k \) estimation.
While some datasets exhibit minimal spread, others display significant variance, suggesting the need for more Monte Carlo simulations.
This suggests a need for more projections on certain datasets.
However, this variance does not significantly impede \( k \) estimation.
Summarizing, Figure~\ref{fig:relative_k_diagram} substantiates the notion that increasing the number of Monte Carlo simulations augments the precision of the mp-means algorithm's k estimation, albeit with varying degrees of improvement across different datasets.
Interestingly, we note that convergence occurs at around 100 simulations, which is significant for practical use.

\section{Conclusions}
\label{conclusions}

In this paper, we addressed the challenges associated with generalizing unimodality testing to higher dimensions.
Building upon the notion of $\alpha$-unimodality, we presented a novel methodology that provides a robust approach to multivariate unimodality testing.
Utilizing point-to-point distancing and linear random projections, our approach bridges the gap between the simplicity of one-dimensional unimodality confirmation and the intricacies of its higher-dimensional counterpart.
By integrating our unimodality test with k-means, we introduced a new incremental clustering approach that automatically estimates the cluster count while performing clustering.
Empirical evaluations affirm our test's efficacy in unimodality assessments and clustering scenarios, highlighting its vital applicability across diverse data-driven domains.
Future exploration will focus on identifying additional \( \alpha \)-unimodality preserving operations to enhance our test, e.g, exploiting JL lemma's relaxations through matrix sketching, experimenting with \( \alpha \)-unimodality for \( \alpha \neq 1 \), and applying the test to various data analysis tasks, such as covariate shift and time series change detection.

\bibliography{main}
\bibliographystyle{icml2024}

\newpage
\appendix
\onecolumn

\begin{table*}[!h]
    \centering
    \caption{Datasets for our experiments: Size indicates data instances, Dimension shows original encoding's flattened size, and k represents ground-truth class labels.
     \vspace{0.1in}
    }
    \begin{tabular}{l l l l l l l}
    \toprule
    \textbf{Dataset} & \textbf{Type} & \textbf{Description} & \textbf{Size} & \textbf{Dimension} & \textbf{k} & \textbf{Source} \\
    \midrule
    USPS & Image & Handwritten digits & 10000 & 256 & 10 & \citet{hull1994database} \\
    MNIST & Image & Handwritten digits & 10000 & 784 & 10 & \citet{lecun1998mnist} \\
    F-MNIST & Image & Zalando's article image & 10000 & 784 & 10 & \citet{f-mnist} \\
    HAR & Time-series & Smartphone-based activity & 2947 & 561 & 5 & \citet{uci} \\
    Optdigits & Image & Handwritten digits & 1797 & 64 & 10 & \citet{uci} \\
    Pendigits & Time-series & Handwritten digits & 10992 & 16 & 10 & \citet{uci} \\
    Isolet & Spectral & Speech recordings pronouncing letters & 6238 & 617 & 26 & \citet{uci} \\
    TCGA & Tabular & Cancer gene expression profiles & 801 & 20531 & 5 & \citet{uci} \\
    \bottomrule
    \end{tabular}
    \label{tab:datasets}
  \end{table*}

  \section{Clustering Datasets}
  \label{appendix:clust_datasets}
Table~\ref{tab:datasets} summarizes the benchmark datasets used in our experiments, varying in size, dimensions, number of classes \(k\) (considered as ground-truth number of clusters), complexity, and domain.
The datasets USPS, MNIST, and F-MNIST feature images of handwritten digits and fashion items, respectively.
Specifically, USPS and MNIST contain grayscale images of handwritten digits from 0 to 9, with MNIST having a resolution of 28 x 28 pixels.
In contrast, F-MNIST, or Fashion MNIST, encompasses grayscale images of ten clothing types, also at a resolution of 28 x 28 pixels.
The Human Activity Recognition (HAR) dataset captures data from sensors to classify human activities, such as walking and sitting.
OptDigits and Pendigits both pertain to handwritten digits: OptDigits consists of 8 x 8 resolution images, while Pendigits uses 16-dimensional vectors containing pixel coordinates.
The Isolet dataset is an assemblage of speech recordings, representing the sounds of spoken letters, characterized by vectors with 617 spectral coefficients derived from the speech.
Finally, TCGA is a compendium of gene expression profiles garnered from RNA sequencing of diverse cancer specimens, inclusive of clinical data, normalized counts, gene annotations, and pathways for five cancer types.
For USPS, MNIST, and F-MNIST, we use the test datasets with 10000 points.
For other datasets, we utilize the full set, typically containing fewer than 10000 points, except for Pendigits.
Despite better preliminary results, we we constrained the sizes of USPS, MNIST, and F-MNIST to their test sets to avoid sample number bias.
 
 \section{ Clustering Algorithms}
 \label{appendix:clust_algos}
In this section, we provide an overview of the clustering algorithms employed in our experiments. We focus on methods that automatically estimate the number of clusters.
We benchmark our approach against well-established algorithms, namely x-means \cite{x-means}, g-means \cite{x-means}, pg-means \cite{pg-means}, dip-means \cite{dip-means}, and hdbscan \cite{hdbscan}.
Given our design's adherence to the original dataspace, we exclude methods combining learning embeddings and clustering, like deep clustering approaches.
To select the best number of clusters, x-means incorporates a regularization penalty guided by the Bayesian Information Criterion (BIC), which accounts for model complexity.
However, this approach excels mainly with abundant data and distinct spherical clusters. Another extension, g-means, tests the assumption that each cluster originates from a Gaussian distribution. Given the challenges of statistical tests in high dimensions, g-means first projects cluster datapoints onto a high variance axis and then employs the Anderson-Darling test for normality. Clusters failing this test are iteratively split to identify the Gaussian mixture.
Conversely, projected g-means (pg-means) assumes a Gaussian mixture for the entire dataset, evaluating the model as a whole. It relies on the EM algorithm, constructing one-dimensional projections of both the dataset and the learned model, and subsequently assessing model fit in the projected space using the Kolmogorov-Smirnov (KS) test. This approach's strength lies in identifying overlapping Gaussian clusters of varying scales and covariances.
 
Moreover, hdbscan enhances dbscan by transforming it into a hierarchical clustering method and subsequently employs a technique to derive a flat clustering based on cluster stability.
Hdbscan aims to obtain an optimal cluster solution by maximizing the aggregate stability of chosen clusters. 
Both mp-means and fold-means build upon the wrapper method (dip-means) introduced in the work of \citet{dip-means}.
Dip-means draws upon this approach and uses the dip-dist criterion internally.
These methods incrementally increase cluster count and apply unimodality tests to clusters shaped by k-means.
The process concludes when all clusters are characterized as unimodal.

For baseline model hyperparameters, we adopt a significance level \( a=0.01 \) for all relevant algorithms.
The default setups are used for x-means and g-means from \citet{pyclustering}, pg-means as per \citet{pg-means}, and hdbscan following \citet{hdbscan_implementation}. The default hyperparameters for dip-means are those provided by the authors\footnote{\url{https://kalogeratos.com/psite/material/dip-means/}}.
For the folding test, we utilize the publicly available Python version\footnote{\url{https://github.com/asiffer/python3-libfolding}}.
Given the folding test's propensity to indicate multimodality, we set a cap on \( k \) for fold-means.
Specifically, for all algorithms necessitating a maximum \( k \) value, we set \( k_{\text{max}} = 300 \).

 \section{Implementation Details}
 \label{appendix:impl_details}

The code implementation, developed in Python, utilizes the diptest library\footnote{\url{https://pypi.org/project/diptest/}} for computing dip values and statistics.
The code is publicly available at: \href{https://github.com/prokolyvakis/mudpod}{https://github.com/prokolyvakis/mudpod}.


\end{document}